\documentclass[11pt,a4paper]{article}%
\usepackage{amsfonts}
\usepackage{amsmath}
\usepackage{amssymb}
\usepackage{graphicx}
\usepackage{geometry}%
\providecommand{\U}[1]{\protect\rule{.1in}{.1in}}
\newtheorem{theorem}{Theorem}

\newtheorem{corollary}{Corollary}

\newtheorem{proposition}{Proposition}
\newtheorem{remark}{Remark}

\newenvironment{proof}[1][Proof]{\textbf{#1.} }{\  \rule{0.5em}{0.5em}}
\makeatletter
\def \@removefromreset#1#2{\let \@tempb \@elt
\def \@tempa#1{@&#1}\expandafter \let \csname @*#1*\endcsname \@tempa
\def \@elt##1{\expandafter \ifx \csname @*##1*\endcsname \@tempa \else
\noexpand \@elt{##1}\fi}     \expandafter \edef \csname cl@#2\endcsname{\csname cl@#2\endcsname}     \let \@elt \@tempb
\expandafter \let \csname @*#1*\endcsname \@undefined}

\@removefromreset{equation}{section}

\@removefromreset{theorem}{section}
\makeatother
\begin{document}

\title{New general lower and upper bounds under minimum-error quantum state discrimination}
\author{Elena R. Loubenets\\National Research University Higher School of Economics, \\Moscow 101000, Russia }
\maketitle

\begin{abstract}
For the optimal success probability under minimum-error discrimination between
$r\geq2$ arbitrary quantum states prepared with any a priori probabilities, we
find new general analytical lower and upper bounds and specify the relations
between these new general bounds and the general bounds known in the
literature. We also present the example where the new general analytical
bounds, lower and upper, on the optimal success probability are tighter than
most of the general analytical bounds known in the literature. The new upper
bound on the optimal success probability explicitly generalizes to $r>2$ the
form of the Helstrom bound. For $r=2$, each of our new bounds, lower and
upper, reduces to the Helstrom bound.

\end{abstract}

\section{Introduction}

Different aspects of quantum state discrimination are discussed in the
literature ever since the seminal papers of Helstrom and Holevo
\cite{1,2,3,4,5} and are now presented in many textbooks and reviews, see, for
example, in \cite{6,7,8} and references therein.

Let a sender prepare a quantum system described in terms of a complex Hilbert
space $\mathcal{H}$ in one of $r\geq2$ quantum states $\rho_{1},...,\rho_{r},$
pure or mixed, with probabilities $q_{1},...,q_{r},$ $\sum_{i}q_{i}%
=1,\mathit{\ }q_{i}>0$, and send this quantum system in an initial state
\begin{equation}
\rho=\sum_{i=1,...,r}q_{i}\rho_{i},\mathit{\ \ \ }\sum_{i=1,..,r}%
q_{i}=1,\mathit{\ \ \ }q_{i}>0, \label{1}%
\end{equation}
to a receiver. For discriminating between states $\rho_{1},...,\rho_{r},$ a
receiver performs a measurement described by a POV measure
\begin{equation}
\mathrm{M}_{r}=\left\{  \mathrm{M}_{r}(i),\text{ \ }i=1,...,r;\text{ \ \ }%
\sum_{i=1,...,r}\mathrm{M}_{r}(i)=\mathbb{I}_{\mathcal{H}}\right\}  \label{2}%
\end{equation}
and the success probability to take under this measurement the proper decision
equals to%
\begin{equation}
\mathrm{P}_{\rho_{1},...,\rho_{_{r}}|q_{1},...,q_{_{r}}}^{success}%
(\mathrm{M}_{r})=\sum_{i=1,...,r}q_{i}\mathrm{tr}\left\{  \rho_{i}%
\mathrm{M}_{r}(i)\right\}  , \label{3}%
\end{equation}
correspondingly, the error probability
\begin{align}
\mathrm{P}_{\rho_{1},...,\rho_{_{r}}|q_{1},...,q_{_{r}}}^{error}%
(\mathrm{M}_{r})  &  =\sum_{i=1,...,r}q_{i}\mathrm{tr}\left\{  \rho
_{i}(\mathbb{I}-\mathrm{M}_{r}(i))\right\} \label{3.1}\\
&  =1-\mathrm{P}_{\rho_{1},...,\rho_{_{r}}|q_{1},...,q_{_{r}}}^{success}%
(\mathrm{M}_{r}).\nonumber
\end{align}

Denote by $\mathfrak{M}_{r}=\left\{  \mathrm{M}_{r}\right\}  ,$ $r\geq2,$ the
set of all possible POV measures (\ref{2}). Under the maximum likelihood (the
minimum error) state discrimination strategy, the optimal success probability
and the optimal error probability are given by%
\begin{align}
\mathrm{P}_{\rho_{1},...,\rho_{_{r}}|q_{1},...,q_{_{r}}}^{opt.success}  &
:\text{ }=\max_{\mathrm{M}_{r}\text{ }\in\text{ }\mathfrak{M}_{r}}%
\mathrm{P}_{\rho_{1},...,\rho_{_{r}}|q_{1},...,q_{_{r}}}^{success}%
(\mathrm{M}_{r}),\label{4}\\
\mathrm{P}_{\rho_{1},...,\rho_{_{r}}|q_{1},...,q_{_{r}}}^{opt.error}  &
:\text{ }=\min_{\mathrm{M}_{r}\text{ }\in\text{ }\mathfrak{M}_{r}}%
\mathrm{P}_{\rho_{1},...,\rho_{_{r}}|q_{1},...,q_{_{r}}}^{error}%
(\mathrm{M}_{r})=1-\mathrm{P}_{\rho_{1},...,\rho_{_{r}}|q_{1},...,q_{_{r}}%
}^{opt.success}, \label{4.1}%
\end{align}
and are attained at some extreme point of the convex set $\mathfrak{M}_{r}.$

The alternative expressions for the optimal error probability (\ref{4.1}) are
presented in Theorem 1 and Corollary 1 of \cite{9-1}.

The following general statement was first formulated and proved by Holevo in
\cite{3,4}.

\begin{theorem}
Under the maximum likelihood (the minimum error) state discrimination
strategy, a POV measure $\mathrm{M}_{r}^{(opt)}$ $\in$ $\mathfrak{M}_{r}$ is
optimal if and only if there exists a self-adjoint trace class operator
$\Lambda_{0}$ such that: \textrm{(i)} $\left(  \Lambda_{0}-q_{i}\rho
_{i}\right)  \mathrm{M}_{r}^{(opt)}(i)=0;$ \textrm{(ii)} $\Lambda_{0}\geq
q_{i}\rho_{i},$ for all $i=1,...,r.$ Herewith,%
\begin{equation}
\mathrm{P}_{\rho_{1},...,\rho_{_{r}}|q_{1},...,q_{_{r}}}^{opt.success}%
=\mathrm{tr}\{\Lambda_{0}\},\text{ \ \ \ \ }\Lambda_{0}=\sum_{i=1,...,r}%
q_{i}\rho_{i}\mathrm{M}_{r}^{(opt)}(i). \label{5.1}%
\end{equation}

\end{theorem}

For $r=2,$ the success probability (\ref{3}) admits the Helstrom upper bound
\cite{1,2,5}
\begin{equation}
\mathrm{P}_{\rho_{1},\rho_{2}|q_{1},q_{2}}^{success}(\mathrm{M}_{2})\leq
\frac{1}{2}\left(  1+\left\Vert q_{1}\rho_{1}-q_{2}\rho_{2}\right\Vert
_{1}\right)  , \label{5}%
\end{equation}
which is attained, so that, for $r=2,$ the optimal success probability is
given by \cite{1,2,5}
\begin{equation}
\mathrm{P}_{\rho_{1},\rho_{2}|q_{1},q_{2}}^{opt.success}=\frac{1}{2}\left(
1+\left\Vert q_{1}\rho_{1}-q_{2}\rho_{2}\right\Vert _{1}\right)  , \label{6}%
\end{equation}
where $\left\Vert \cdot\right\Vert _{1}$ is the trace norm.

For an arbitrary $r>2$, a precise general expression for the optimal success
probability (\ref{4}) is not known, however, there were introduced and studied
\cite{10,11,12,13,14,15,16,17,18, 19,20} several general upper and lower
bounds on the optimal success probability $\mathrm{P}_{\rho_{1},...,\rho
_{_{r}}|q_{1},...,q_{_{r}}}^{opt.success}$, expressed via different
characteristics of quantum states.

As proved by Qiu\&Li \cite{16}, in some cases, the general lower bound on the
optimal error probability (correspondingly, the upper bound on the optimal
success probability $\mathrm{P}_{\rho_{1},...,\rho_{_{r}}|q_{1},...,q_{_{r}}%
}^{opt.success}),$ introduced by them in \cite{16}, is tighter than the other
general lower bounds known in the literature.

Computation of bounds on $\mathrm{P}_{\rho_{1},...,\rho_{_{r}}|q_{1}%
,...,q_{_{r}}}^{opt.success}$within semidefinite programming was considered
recently in \cite{20}.

In the present article, for the optimal success probability (\ref{4}), we find
(Theorems 2, 3) for all $r\geq2$ the new general lower and upper bounds and
specify (Propositions 1, 2) the relation of our bounds to the general lower
and upper bounds known in literature. For $r=2$, each of the new general
bounds, lower and upper, reduces to the Helstrom bound in (\ref{5}), and this
proves in the other way the Helstrom result (\ref{6}).

\section{New lower bounds}

Taking into account that $\mathrm{M}_{r}(j)=\mathbb{I}_{\mathcal{H}}%
-\sum_{i\neq j}\mathrm{M}_{r}(i),$ we rewrite the right-hand side of
expression (\ref{3}) in either of $j$-th representations:%
\begin{align}
\mathrm{P}_{\rho_{1},...,\rho_{_{r}}|q_{1},...,q_{_{r}}}^{success}%
(\mathrm{M}_{r})  &  =\sum_{i=1,...,r}q_{i}\mathrm{tr}\{\rho_{i}\mathrm{M}%
_{r}(i)\}\label{8}\\
&  =q_{j}+\sum_{i=1,...,r}\mathrm{tr}\{\left(  q_{i}\rho_{i}-q_{j}\rho
_{j}\right)  \mathrm{M}_{r}(i)\},\text{ \ \ }j=1,...,r,\nonumber
\end{align}
for every POV measure $\mathrm{M}_{r}.$ Summing up the left-hand and the
right-hand sides of (\ref{8}) over all $j=1,...,r$, for any POV measure
$\mathrm{M}_{r}\in\mathfrak{M}_{r},$ we also come to the following
representation for the success probability%
\begin{equation}
\mathrm{P}_{\rho_{1},...,\rho_{_{r}}|q_{1},...,q_{_{r}}}^{success}%
(\mathrm{M}_{r})=\frac{1}{r}\left(  1+\sum_{i,j=1,...,r}\mathrm{tr}\left\{
\left(  q_{i}\rho_{i}-q_{j}\rho_{j}\right)  \mathrm{M}_{r}(i)\text{ }\right\}
\right)  . \label{9}%
\end{equation}

Recall that a self-adjoint (Hermitian) bounded operator $X$ on $\mathcal{H}$
admits the decomposition:%
\begin{align}
X  &  =X^{(+)}-X^{(-)},\text{ \ \ }X^{(\pm)}\geq0,\label{10}\\
X^{(+)}  &  =\sum_{\lambda_{k}>0,}\lambda_{k}\mathrm{E}_{X}(\lambda
_{k}),\text{ \ \ }X^{(-)}=\sum_{\lambda_{k}\leq0,}\left\vert \lambda
_{k}\right\vert \mathrm{E}_{X}(\lambda_{k}),\nonumber
\end{align}
where $\mathrm{E}_{X}(\lambda_{k})$ the spectral projections of a Hermitian
operator $X$. If a bounded operator $X$ is trace class, then operators
$X^{(\pm)}\geq0$ are also trace class and
\begin{align}
\left\Vert X\right\Vert _{1}  &  :=\mathrm{tr}\left\vert X\right\vert ,\text{
\ \ }\left\vert X\right\vert =X^{(+)}+X^{(-)},\label{11}\\
\left\Vert X^{(\pm)}\right\Vert _{1}  &  =\mathrm{tr}\{X^{(\pm)}\}.\nonumber
\end{align}
From relations \cite{7} $\left\vert \mathrm{tr}\{W\}\right\vert \leq\left\Vert
W\right\Vert _{1}$ and $\left\Vert AB\right\Vert _{1}\leq\left\Vert
A\right\Vert _{1}\left\Vert B\right\Vert _{0}$, valid for all trace-class
operators $W,$ $A$ and all bounded operators $B,$ it follows that if
$X,Y$\ $\geq0$ (hence, $\mathrm{tr}\{XY\}\geq0),$ then%
\begin{equation}
0\leq\mathrm{tr}\{XY\}\leq\left\Vert X\right\Vert _{1}\left\Vert Y\right\Vert
_{0}, \label{12}%
\end{equation}
where notation $\left\Vert \cdot\right\Vert _{0}$ means the operator norm.

Definition (\ref{4}) and\ relations (\ref{8})--(\ref{12}) imply.

\begin{theorem}
[New lower bounds]For any number $r\geq2$ of arbitrary quantum states
$\rho_{1},...,\rho_{r}$ prepared with probabilities $q_{1},...,q_{r}$, the
optimal success probability (\ref{4}) admits the lower bounds%
\begin{align}
\mathrm{P}_{\rho_{1},...,\rho_{r}|q_{1},...,q_{r}}^{opt.success}  &
\geq\mathfrak{L}_{1,new}^{(r)}:=\max_{j=1,...,r}\left\{  q_{j}+\frac{1}%
{r-1}\sum_{i=1,...,r}\left\Vert \left(  q_{i}\rho_{i}-q_{j}\rho_{j}\right)
^{(+)}\right\Vert _{1}\right\} \label{15}\\
&  =\frac{1}{2(r-1)}+\frac{1}{2(r-1)}\max_{j=1,...,r}\left\{  \sum
_{i=1,...,r}\left\Vert q_{i}\rho_{i}-q_{j}\rho_{j}\right\Vert _{1}%
+q_{j}(r-2)\right\} \label{16}\\
&  \geq\mathfrak{L}_{2,new}^{(r)}:=\frac{1}{r}\left(  1+\frac{1}{r-1}%
\sum_{1\leq i<j\leq r}\left\Vert q_{i}\rho_{i}-q_{j}\rho_{j}\right\Vert
_{1}\right)  . \label{17}%
\end{align}
For $r=2,$ each of these new lower bounds reduces to the Helstrom bound in
(\ref{5}).
\end{theorem}

\begin{proof}
Let $\mathrm{E}_{(q_{i}\rho_{i}-q_{j}\rho_{j})}(\lambda_{k})$ be the spectral
projections of the Hermitian operator $(q_{i}\rho_{i}-q_{j}\rho_{j})$ on
$\mathcal{H}$ and
\begin{equation}
\mathrm{P}_{ij}^{(+)}:=\sum_{\lambda_{k}>0}\mathrm{E}_{(q_{i}\rho_{i}%
-q_{j}\rho_{j})}(\lambda_{k}),\text{ \ \ }i\neq j, \label{18}%
\end{equation}
denote the orthogonal projection on the proper subspace of operator
$(q_{i}\rho_{i}-q_{j}\rho_{j}),$ corresponding to its positive eigenvalues.
Note that by (\ref{11}):%
\begin{align}
q_{i}-q_{j}  &  =\left\Vert \left(  q_{i}\rho_{i}-q_{j}\rho_{j}\right)
^{(+)}\right\Vert _{1}-\left\Vert \left(  q_{i}\rho_{i}-q_{j}\rho_{j}\right)
^{(-)}\right\Vert _{1},\label{20.1}\\
\left\Vert q_{i}\rho_{i}-q_{j}\rho_{j}\right\Vert _{1}  &  =\left\Vert \left(
q_{i}\rho_{i}-q_{j}\rho_{j}\right)  ^{(+)}\right\Vert _{1}+\left\Vert \left(
q_{i}\rho_{i}-q_{j}\rho_{j}\right)  ^{(-)}\right\Vert _{1}.\nonumber
\end{align}
Introduce the POV measures $\mathrm{M}_{r}^{(j)},$ $j=1,....,r,$ each with the
elements%
\begin{align}
\mathrm{M}_{r}^{(j)}(i)  &  =\frac{1}{r-1}\mathrm{P}_{ij}^{(+)},\text{
\ \ }\ i\neq j.\label{19}\\
\mathrm{M}_{r}^{(j)}(j)  &  =\mathbb{I}_{\mathcal{H}}-\frac{1}{r-1}%
\sum_{i=1,...,r,\text{ }i\neq j}\mathrm{P}_{ij}^{(+)}\nonumber\\
&  =\frac{1}{r-1}\sum_{i=1,...,r,\text{ }i\neq j}(\mathbb{I}_{\mathcal{H}%
}-\mathrm{P}_{ij}^{(+)})\geq0.\nonumber
\end{align}
From the $j$-th representation in (\ref{8}) and relations (\ref{20.1}) it
follows that, for the $j$-th POV measure (\ref{19}), we have:%
\begin{align}
\mathrm{P}_{\rho_{1},...,\rho_{r}|q_{1},...,q_{r}}^{success}(\mathrm{M}%
_{r}^{(j)})  &  =q_{j}+\sum_{i=1,...,N}\mathrm{tr}\{\left(  q_{i}\rho
_{i}-q_{j}\rho_{j}\right)  \mathrm{M}_{r}^{(j)}(i)\}\label{20}\\
&  =q_{j}+\frac{1}{r-1}\sum_{i=1,...,r}\left\Vert \left(  q_{i}\rho_{i}%
-q_{j}\rho_{j}\right)  ^{(+)}\right\Vert _{1},\nonumber\\
j  &  =1,...,r.\nonumber
\end{align}
For the optimal success probability (\ref{4}), equalities (\ref{20}) imply%
\begin{align}
\mathrm{P}_{\rho_{1},...,\rho_{r}|q_{1},...,q_{r}}^{opt.success}  &  \geq
q_{j}+\frac{1}{r-1}\sum_{i=1,...,r}\left\Vert \left(  q_{i}\rho_{i}-q_{j}%
\rho_{j}\right)  ^{(+)}\right\Vert _{1},\label{21}\\
\forall j  &  =1,...,r,\nonumber
\end{align}
hence,
\begin{equation}
\mathrm{P}_{\rho_{1},...,\rho_{r}|q_{1},...,q_{r}}^{opt.success}\geq
\max_{j=1,...,r}\left\{  q_{j}+\frac{1}{r-1}\sum_{i=1,...,r}\left\Vert \left(
q_{i}\rho_{i}-q_{j}\rho_{j}\right)  ^{(+)}\right\Vert _{1}\right\}  .
\label{22}%
\end{equation}
Since by (\ref{20.1})
\begin{equation}
\left\Vert \left(  q_{i}\rho_{i}-q_{j}\rho_{j}\right)  ^{(+)}\right\Vert
_{1}=\frac{1}{2}(q_{i}-q_{j})+\frac{1}{2}\left\Vert q_{i}\rho_{i}-q_{j}%
\rho_{j}\right\Vert _{1}, \label{23}%
\end{equation}
the expression in the right-hand side of (\ref{21}) is otherwise equal to%
\begin{align}
&  q_{j}+\frac{1}{r-1}\sum_{i=1,...,r}\left\Vert \left(  q_{i}\rho_{i}%
-q_{j}\rho_{j}\right)  ^{(+)}\right\Vert _{1}\label{24}\\
&  =\frac{1}{2(r-1)}+\frac{1}{2(r-1)}\left\{  \sum_{i=1,...,r}\left\Vert
q_{i}\rho_{i}-q_{j}\rho_{j}\right\Vert _{1}+q_{j}(r-2)\right\}  .\nonumber
\end{align}
This and relation (\ref{22}) imply the lower bounds (\ref{15}) and (\ref{16}).
Summing up the left-hand and the right-hand sides of (\ref{21}) over all
$j=1,...,r$ and taking into account $\sum_{j=1,...,r}q_{j}=1,$ relations
(\ref{20.1}) and
\begin{equation}
\left(  q_{i}\rho_{i}-q_{j}\rho_{j}\right)  ^{+}=\left(  q_{j}\rho_{j}%
-q_{i}\rho_{i}\right)  ^{(-)}, \label{25}%
\end{equation}
we derive%
\begin{equation}
\mathrm{P}_{\rho_{1},...,\rho_{r}|q_{1},...,q_{r}}^{opt.success}\geq\frac
{1}{r}\left(  1+\frac{1}{r-1}\sum_{1\leq i<j\leq r}\left\Vert q_{i}\rho
_{i}-q_{j}\rho_{j}\right\Vert _{1}\right)  , \label{26}%
\end{equation}
that is, the lower bound (\ref{17}). Furthermore, since, for any positive
numbers $\alpha_{j},$ $j=1,...,r,$ their sum
\begin{equation}
\sum_{j=1,...,r}\alpha_{j}\leq r\max_{j=1,...,r}\alpha_{j}, \label{27}%
\end{equation}
we have%
\begin{align}
&  \max_{j=1,...,r}\left\{  q_{j}+\frac{1}{r-1}\sum_{i=1,...,N}\left\Vert
\left(  q_{i}\rho_{i}-q_{j}\rho_{j}\right)  ^{(+)}\right\Vert _{1}\right\}
\label{28}\\
&  \geq\frac{1}{r}\left(  1+\frac{1}{r-1}\sum_{1\leq i<j\leq r}\left\Vert
q_{i}\rho_{i}-q_{j}\rho_{j}\right\Vert _{1}\right)  .\nonumber
\end{align}
Relations (\ref{22}), (\ref{24}) and (\ref{28}) prove the statement of Theorem 2.
\end{proof}

Consider now the relation of the new lower bounds (\ref{15})--(\ref{17}) to
the known general lower bounds on $\mathrm{P}_{\rho_{1},...,\rho_{r}%
|q_{1},...,q_{r}}^{opt.success}:$%
\begin{align}
\mathfrak{L}_{1}^{(r)}  &  :=\max_{j=1,...,r}q_{j}\geq\frac{1}{r}%
,\label{28.1}\\
\mathfrak{L}_{2}^{(r)}  &  :=1-\sum_{1\leq i<j\leq r}\sqrt{q_{i}q_{j}}%
F_{ij},\label{28.2}\\
\mathfrak{L}_{3}^{(r)}  &  :=\left(  \mathrm{tr}\left[  \sqrt{%
{\textstyle\sum\limits_{i=1,...,r}}
q_{i}^{2}\rho_{i}^{2}}\right]  \right)  ^{2}, \label{28.3}%
\end{align}
where (a) bound (\ref{28.1}) follows from item (ii) and relation (\ref{5.1})
in Theorem 1; (b) bound (\ref{28.2}) was introduced by Barnum\&Knill in
\cite{9} and further studied by Audenaert\&Mosonyi in \cite{19}; (c) bound
(\ref{28.3}) was introduced by Tyson in \cite{18}. Here, $F_{ij}:=\left\Vert
\sqrt{\rho_{i}}\sqrt{\rho_{j}}\right\Vert _{1}$ is the pairwise fidelity.

\begin{proposition}
(i) The new lower bound (\ref{15}) is tighter%
\begin{equation}
\mathfrak{L}_{1,new}^{(r)}\geq\mathfrak{L}_{1}^{(r)} \label{29.1}%
\end{equation}
than the known lower bound (\ref{28.1}) for any number $r\geq2$ of arbitrary
states $\rho_{1},...,\rho_{r}$ and probabilities $q_{1},...,q_{r}$.
\newline(ii) For $r=2,$ the new lower bounds (\ref{15}), (\ref{17}) are equal
to the Helstrom bound and $\mathfrak{L}_{1,new}^{(2)}=\mathfrak{L}%
_{2,new}^{(2)}\geq\mathfrak{L}_{2}^{(2)}$ for all states $\rho_{1},\rho_{2}$
and probabilities $q_{1},q_{2}.$ \newline(iii) For any $r>2,$ the new lower
bound (\ref{17}) is tighter than the lower bound (\ref{28.2})%
\begin{equation}
\mathfrak{L}_{2,new}^{(r)}\geq\mathfrak{L}_{2}^{(r)} \label{29.3}%
\end{equation}
if%
\begin{equation}
\sum_{1\leq i<j\leq r}\left\Vert q_{i}\rho_{i}-q_{j}\rho_{j}\right\Vert
_{1}\leq\frac{(r-1)^{2}}{r+1} \label{29.4}%
\end{equation}
and otherwise if\
\begin{equation}
\frac{(r-1)^{2}}{r+1}<\sum_{1\leq i<j\leq r}\left\Vert q_{i}\rho_{i}-q_{j}%
\rho_{j}\right\Vert _{1}\leq r-1. \label{29.5}%
\end{equation}

\end{proposition}

\begin{proof}
Due to the structure of the new lower bound (\ref{15}), relation (\ref{29.1})
is obvious. In order to prove (ii) and (iii), consider the difference%
\begin{equation}
\mathfrak{L}_{2,new}^{(r)}-\mathfrak{L}_{2}^{(r)}=\frac{1}{r}+\frac{1}%
{r(r-1)}\sum_{1\leq i<j\leq r}\left\Vert q_{i}\rho_{i}-q_{j}\rho
_{j}\right\Vert _{1}-1+\sum_{1\leq i<j\leq r}\sqrt{q_{i}q_{j}}F(\rho_{i}%
,\rho_{j}). \label{29.6}%
\end{equation}
By Lemma 5 in \cite{16}
\begin{equation}
\sqrt{q_{i}q_{j}}F(\rho_{i},\rho_{j})\geq\frac{1}{2}(q_{i}+q_{j})-\frac{1}%
{2}\left\Vert q_{i}\rho_{i}-q_{j}\rho_{j}\right\Vert _{1}. \label{29.7}%
\end{equation}
Substituting this relation into (\ref{29.6}) and taking into account
\begin{equation}
\sum_{1\leq i<j\leq r}(q_{i}+q_{j})=r-1, \label{29.8}%
\end{equation}
we derive%
\begin{align}
\mathfrak{L}_{2,new}^{(r)}-\mathfrak{L}_{2}^{(r)}  &  \geq\frac{(r-2)(r-1)}%
{2r}-\left(  \frac{1}{2}-\frac{1}{r(r-1)}\right)  \sum_{1\leq i<j\leq
r}\left\Vert q_{i}\rho_{i}-q_{j}\rho_{j}\right\Vert _{1}\label{29.9}\\
&  =\frac{(r-2)(r-1)}{2r}\left(  1-\frac{r+1}{(r-1)^{2}}\sum_{1\leq i<j\leq
r}\left\Vert q_{i}\rho_{i}-q_{j}\rho_{j}\right\Vert _{1}\right)  .\nonumber
\end{align}
This proves relations (\ref{29.3}), (\ref{29.4}), also, the left-hand side
inequality of (\ref{29.5}). The right-hand side inequality of (\ref{29.5})
follows from the relation%
\begin{equation}
\sum_{1\leq i<j\leq r}\left\Vert q_{i}\rho_{i}-q_{j}\rho_{j}\right\Vert
_{1}\leq\sum_{1\leq i<j\leq r}(q_{i}+q_{j})=r-1, \label{29.10}%
\end{equation}
where we took into account (\ref{29.8}).
\end{proof}

\section{New upper bound}

From relations (\ref{10})--(\ref{12}) and inequality $\left\Vert
\mathrm{M}_{r}(i)\right\Vert _{0}\leq1$ it follows that, for every POV measure
$\mathrm{M}_{r},$ in representation (\ref{8})%
\begin{align}
\mathrm{P}_{\rho_{1},...,\rho_{_{r}}|q_{1},...,q_{_{r}}}^{success}%
(\mathrm{M}_{r})  &  =q_{k}+\sum_{i=1,...,r}\mathrm{tr}\{(q_{i}\rho_{i}%
-q_{k}\rho_{k})\mathrm{M}_{r}(i)\}\label{13}\\
&  \leq q_{m}+\sum_{i=1,...,r}\left\Vert \left(  q_{i}\rho_{i}-q_{m}\rho
_{m}\right)  ^{(+)}\right\Vert _{1},\text{ \ \ \ }k,m=1,...,r,\nonumber
\end{align}
and in representation (\ref{9})
\begin{align}
\mathrm{P}_{\rho_{1},...,\rho_{_{_{r}}}|q_{1},...,q_{_{r}}}^{success}%
(\mathrm{M}_{r})  &  =\frac{1}{r}\left(  1+\sum_{i,j=1,.,,,r}\mathrm{tr}%
\left\{  \text{ }\left(  q_{i}\rho_{i}-q_{j}\rho_{j}\right)  \mathrm{M}%
_{r}(i)\right\}  \right) \label{14}\\
&  \leq\frac{1}{r}\left(  1+\sum_{i,j=1,...,r}\left\Vert \left(  q_{i}\rho
_{i}-q_{j}\rho_{j}\right)  ^{(+)}\right\Vert _{1}\right)  .\nonumber
\end{align}

Relation (\ref{13}) immediately implies the upper bound%
\begin{align}
\mathrm{P}_{\rho_{1},...,\rho_{_{r}}|q_{1},...,q_{_{r}}}^{opt.success}  &
\leq Q_{4}^{(r)}:=\min_{j=1,...,r}\left\{  q_{j}+\sum_{i=1,...,r}\left\Vert
\left(  q_{i}\rho_{i}-q_{j}\rho_{j}\right)  ^{(+)}\right\Vert _{1}\right\}
\label{30}\\
&  =\frac{1}{2}+\frac{1}{2}\min_{j=1,...,r}\left\{  \sum_{i=1,...,r}\left\Vert
q_{i}\rho_{i}-q_{j}\rho_{j}\right\Vert _{1}-q_{j}(r-2)\right\}  ,\nonumber
\end{align}
which agrees due to the relation $\mathrm{P}_{\rho_{1},...,\rho_{_{r}}%
|q_{1},...,q_{_{r}}}^{opt.success}=1-$ $\mathrm{P}_{\rho_{1},...,\rho_{_{r}%
}|q_{1},...,q_{_{r}}}^{opt.error}$ with the lower bound $L_{4}$ by Qiu\&Li on
$\mathrm{P}_{\rho_{1},...,\rho_{_{r}}|q_{1},...,q_{_{r}}}^{opt.error}$
introduced in \cite{16}. Here, in order to derive the expression in the second
line of (\ref{30}) we took into account relation (\ref{23}).

For convenience in comparison, we take for the upper bound (\ref{30}) and the
below upper bounds (\ref{40} )-(\ref{42}) on the optimal success probability
the numeration similar to that for the lower bounds $L_{n}^{(r)}$ in \cite{16}
on the optimal error probability with the obvious correspondence $Q_{n}%
^{(r)}=1-L_{n}^{(r)}.\medskip$

The following theorem introduces a new upper bound on the optimal success
probability $\mathrm{P}_{\rho_{1},...,\rho_{_{r}}|q_{1},...,q_{_{r}}%
}^{opt.success}$and establishes its relation to the upper bound (\ref{30}).

\begin{theorem}
[New upper bound]For any number $r\geq2$ of arbitrary quantum states $\rho
_{1},...,\rho_{r_{r}}$ prepared with probabilities $q_{1},...,q_{_{r}},$ the
optimal success probability (\ref{4}) admits the upper bound%
\begin{equation}
\mathrm{P}_{\rho_{1},...,\rho_{r}|q_{1},...,q_{r}}^{opt.success}\leq
Q_{new}^{(r)}:=\frac{1}{r}\left(  1+\sum_{1\leq i<j\leq r}\left\Vert q_{i}%
\rho_{i}-q_{j}\rho_{j}\right\Vert _{1}\right)  , \label{32}%
\end{equation}
explicitly generalizing to $r>2$ the form of the Helstrom upper bound in
(\ref{5}) and relating to the upper bound (\ref{30}) as%
\begin{equation}
Q_{4}^{(r)}\leq Q_{new}^{(r)}. \label{33}%
\end{equation}

\end{theorem}

\begin{proof}
In view of (\ref{4}), relations (\ref{14}) and (\ref{25}) immediately imply
the upper bound (\ref{32}). Also, by (\ref{13})
\begin{equation}
\mathrm{P}_{\rho_{1},...,\rho_{r}|q_{1},...,q_{r}}^{opt.success}\leq
q_{j}+\sum_{i=1,...,r}\left\Vert \left(  q_{i}\rho_{i}-q_{j}\rho_{j}\right)
^{(+)}\right\Vert _{1},\text{ \ \ \ }j=1,...,r. \label{34}%
\end{equation}
Summing up the left-hand and the right-hand sides of this relation over
$j=1,...,r$, and taking into account (\ref{25}), we derive%
\begin{align}
r\mathrm{P}_{\rho_{1},...,\rho_{r}|q_{1},...,q_{r}}^{opt.success}  &  \leq
\sum_{j=1,...,r}q_{j}+\sum_{i=1,...,r}\left\Vert \left(  q_{i}\rho_{i}%
-q_{j}\rho_{j}\right)  ^{(+)}\right\Vert _{1}\label{35}\\
&  =1+\sum_{1\leq i<j\leq r}\left\Vert q_{i}\rho_{i}-q_{j}\rho_{j}\right\Vert
_{1}.\nonumber
\end{align}
Since, for any positive numbers $\alpha_{j},$ $j=1,...,r,$ their sum
\begin{equation}
\sum_{j=1,...,r}\alpha_{j}\geq r\min_{j=1,...,r}\alpha_{j}, \label{36}%
\end{equation}
we have
\begin{align}
&  \min_{j=1,...,r}\left\{  q_{j}+\sum_{i}\left\Vert \left(  q_{i}\rho
_{i}-q_{j}\rho_{j}\right)  ^{(+)}\right\Vert _{1}\right\} \label{37}\\
&  \leq\frac{1}{r}\left(  1+\sum_{1\leq i<j\leq r}\left\Vert q_{i}\rho
_{i}-q_{j}\rho_{j}\right\Vert _{1}\right)  ,\nonumber
\end{align}
that is, relation (\ref{33}). This proves the statement of Theorem 3.
\end{proof}

We stress that if $r$ is rather large, then the calculation of bound
(\ref{32}) is easier than finding the maximum in bound (\ref{30}).

\begin{remark}
From relations (\ref{23}) and $\left\Vert q_{i}\rho_{i}-q_{j}\rho
_{j}\right\Vert _{1}\leq q_{i}+q_{j},$ it follows
$\vert$%
$\vert$%
$\left(  q_{i}\rho_{i}-q_{j}\rho_{j}\right)  ^{(+)}||_{1}\leq q_{i},$
therefore,%
\begin{equation}
\sum_{i=1,...,r}\left\Vert \left(  q_{i}\rho_{i}-q_{j}\rho_{j}\right)
^{(+)}\right\Vert _{1}\leq\sum_{i=1,...,r,\text{ }i\neq j}q_{i}=1-q_{j}.
\label{38_}%
\end{equation}
Relations (\ref{25}), (\ref{29.8}) and (\ref{38_}) imply
\begin{equation}
\sum_{1\leq i<j\leq r}\left\Vert q_{i}\rho_{i}-q_{j}\rho_{j}\right\Vert
_{1}=\sum_{i,j=1,...,r,\text{ }}\left\Vert \left(  q_{i}\rho_{i}-q_{j}\rho
_{j}\right)  ^{(+)}\right\Vert _{1}\leq r-1, \label{38}%
\end{equation}
which agrees with (\ref{29.10}). Therefore, the new upper bound in Theorem 3
is nontrivial (i. e. not more than one). Also, by relation (\ref{13})
specified for the POV measure (\ref{19}), the bounds in Theorem 2 and the
upper bound (\ref{30}) are consistent in the sense
\begin{align}
&  \max_{k=1,...,r}\left\{  q_{k}+\frac{1}{r-1}\sum_{i=1,...,r}\left\Vert
\left(  q_{i}\rho_{i}-q_{k}\rho_{k}\right)  ^{(+)}\right\Vert _{1}\right\}
\label{39}\\
&  \leq\min_{m=1,...,r}\left\{  q_{m}+\sum_{i=1,...,r}\left\Vert \left(
q_{i}\rho_{i}-q_{m}\rho_{m}\right)  ^{(+)}\right\Vert _{1}\right\} \nonumber
\end{align}
for all $r\geq2.$
\end{remark}

Besides relation (\ref{33}) of the new upper bound (\ref{32}) to the upper
bound (\ref{30}) introduced in \cite{16}, consider also its relation to the
general upper bounds on $\mathrm{P}_{\rho_{1},...,\rho_{r}|q_{1},...,q_{r}%
}^{opt.success}$:%
\begin{align}
Q_{2}^{(r)}  &  :=\frac{1}{2}\left(  1+\frac{1}{r-1}\sum_{1\leq i<j\leq
r}\left\Vert q_{i}\rho_{i}-q_{j}\rho_{j}\right\Vert _{1}\right)  ,\label{40}\\
Q_{3}^{(r)}  &  :=1-\sum_{1\leq i<j\leq r}q_{i}q_{j}F_{ij}^{2},\label{41}\\
Q_{5}^{(r)}  &  :=\mathrm{tr}\left[  \sqrt{%
{\textstyle\sum\limits_{i=1,...,r}}
q_{i}^{2}\rho_{i}^{2}}\right]  , \label{42}%
\end{align}
which follow from the lower bounds $L_{n}^{(r)}$ on the optimal error
probability introduced, correspondingly: (i) by Qiu in \cite{14}; (ii) by
Montanaro in \cite{15}; (iii) by Ogawa\&Nagaoka in \cite{17} for the
equiprobable case and by Tyson \cite{18} for a general case.

The detailed study of these general bounds and some other bounds for specific
families of quantum states is presented in \cite{19}.

\begin{proposition}
(i) For any number $r\geq2,$ of arbitrary quantum states $\rho_{1}%
,...,\rho_{r}$ and a priori probabilities $q_{1},...,q_{r}$, the new upper
bound (\ref{32}) and the upper bounds (\ref{40})--(\ref{41}) satisfy the
relations%
\begin{equation}
Q_{new}^{(r)}\leq Q_{2}^{(r)}\leq1-\frac{1}{N-1}(1-Q_{3}^{(r)}). \label{44}%
\end{equation}
\newline(ii) In the equiprobable case, the new upper bound (\ref{32}) is
tighter than the known upper bound (\ref{41}):
\begin{equation}
Q_{new}^{(r)}\leq Q_{3}^{(r)}. \label{44.1}%
\end{equation}

\end{proposition}

\begin{proof}
(i) The first inequality in (\ref{44}) follows from
\begin{align}
&  Q_{2}^{(r)}-Q_{new}^{(r)}\label{46}\\
&  =\frac{r-2}{2r}\left(  1-\frac{1}{r-1}\sum_{1\leq i<j\leq r}\left\Vert
q_{i}\rho_{i}-q_{j}\rho_{j}\right\Vert _{1}\right)  \geq0,\nonumber
\end{align}
where we take into account estimate (\ref{38}). To prove the second inequality
in (\ref{44}), we use estimate (55) in \cite{15} and derive%
\begin{equation}
1-Q_{2}^{(r)}\geq\frac{1-Q_{3}^{(r)}}{r-1}. \label{47}%
\end{equation}
(ii) For the equiprobable case $q_{1}=...=q_{r}=\frac{1}{r}$, we consider the difference%

\begin{align}
Q_{new}^{(r)}-Q_{3}^{(r)}  &  =\frac{1}{r}+\frac{1}{r^{2}}\sum_{1\leq i<j\leq
r}\left\Vert \rho_{i}-\rho_{j}\right\Vert _{1}-1+\frac{1}{r^{2}}\sum_{1\leq
i<j\leq r}F^{2}(\rho_{i},\rho_{j})\label{47.1}\\
&  =\frac{1}{r^{2}}\sum_{1\leq i<j\leq r}\left(  F^{2}(\rho_{i},\rho
_{j})+\left\Vert \rho_{i}-\rho_{j}\right\Vert _{1}\right)  -\frac{r-1}%
{r}.\nonumber
\end{align}
Taking into account relation (20) in \cite{16} which implies%
\begin{equation}
\sum_{1\leq i<j\leq r}\left(  F^{2}(\rho_{i},\rho_{j})+\left\Vert \rho
_{i}-\rho_{j}\right\Vert _{1}\right)  \leq\sum_{1\leq i<j\leq r}2=r(r-1)
\label{47.2}%
\end{equation}
and substituting this into (\ref{47.1}), we have%
\begin{equation}
Q_{new}^{(r)}-Q_{3}^{(r)}\leq\frac{r-1}{r}-\frac{r-1}{r}=0. \label{47.3}%
\end{equation}
This and relation (\ref{46}) prove the statement.
\end{proof}

\section{General relations}

Theorem 2 and Theorem 3 imply. the following general relations.

\begin{corollary}
For any number $r\geq2$ of arbitrary quantum states $\rho_{1},...,\rho_{_{r}}$
prepared with any probabilities $q_{1},...,q_{_{r}},$ the optimal success
probability (\ref{4}) admits the bounds:%
\begin{align}
&  \frac{1}{r}\left(  1+\frac{1}{(r-1)}\sum_{1\leq i<j\leq r}\left\Vert
q_{i}\rho_{i}-q_{j}\rho_{j}\right\Vert _{1}\right) \nonumber\\
&  \leq\frac{1}{2(r-1)}+\frac{1}{2(r-1)}\max_{j=1,...,r}\left\{
\sum_{i=1,...,r}\left\Vert q_{i}\rho_{i}-q_{j}\rho_{j}\right\Vert _{1}%
+q_{j}(r-2)\right\} \nonumber\\
&  \leq\mathrm{P}_{\rho_{1},...,\rho_{r}|q_{1},...,q_{r}}^{opt.success}%
\label{49}\\
&  \leq\frac{1}{2}+\frac{1}{2}\min_{j=1,...,r}\left\{  \sum_{i=1,...,r}%
\left\Vert q_{i}\rho_{i}-q_{j}\rho_{j}\right\Vert _{1}-q_{j}(r-2)\right\}
\nonumber\\
&  \leq\frac{1}{r}\left(  1+\sum_{1\leq i<j\leq r}\left\Vert q_{i}\rho
_{i}-q_{j}\rho_{j}\right\Vert _{1}\right)  .\nonumber
\end{align}
In the first and the last lines of (\ref{49}), the bounds are quite similar by
its form to the Helstrom upper bound (\ref{5}). For $r=2$, each of the lower
and all bounds in (\ref{49}) reduces to the Helstrom bound in (\ref{5}) and
this proves in the other way the Helstrom result (\ref{6}).
\end{corollary}

For the equiprobable case, bounds (\ref{49}) take the forms.

\begin{corollary}
For any number $r\geq2,$ of arbitrary quantum states $\rho_{1},...,\rho_{r},$
prepared with equal probabilities $q_{1}=...=q_{1}=\frac{1}{r},$ the optimal
success probability (\ref{4}) admits the bounds:%
\begin{align}
&  \frac{1}{r}\left(  1+\frac{1}{r(r-1)}\sum_{1\leq i<j\leq r}\left\Vert
\rho_{i}-\rho_{j}\right\Vert _{1}\right) \nonumber\\
&  \leq\frac{1}{r}+\frac{1}{2r(r-1)}\max_{j=1,...,r}\left\{  \sum
_{i=1,...,r}\left\Vert \rho_{i}-\rho_{j}\right\Vert _{1}\right\} \nonumber\\
&  \leq\mathrm{P}_{\rho_{1},...,\rho_{_{r}}|\frac{1}{r},...,\frac{1}{r}%
}^{opt.success}\label{50}\\
&  \leq\frac{1}{r}+\frac{1}{2r}\min_{j=1,...,r}\left\{  \sum_{i=1,...,r}%
\left\Vert \rho_{i}-\rho_{j}\right\Vert _{1}\right\} \nonumber\\
&  \leq\frac{1}{r}\left(  1+\frac{1}{r}\sum_{1\leq i<j\leq r}\left\Vert
\rho_{i}-\rho_{j}\right\Vert _{1}\right)  .\nonumber
\end{align}

\end{corollary}

\subsection{Example}

For the numerical comparison of the new lower bound (\ref{15})--(\ref{17}) and
the new upper bound (\ref{32}) with the known lower bounds (\ref{28.1}%
)--(\ref{28.3}) and the known upper bounds (\ref{30}), (\ref{40})--(\ref{42}),
we analyze the discrimination between the following three equiprobable qubit
states
\begin{align}
\rho_{1}  &  =\frac{7}{8}\left\vert 0\right\rangle \left\langle 0\right\vert
+\frac{1}{8}\left\vert 1\right\rangle \left\langle 1\right\vert ,\text{
\ \ }\rho_{2}=\frac{5}{8}\left\vert 0\right\rangle \left\langle 0\right\vert
+\frac{3}{8}\left\vert 1\right\rangle \left\langle 1\right\vert ,\label{51}\\
\rho_{3}  &  =\frac{3}{4}\left\vert 0\right\rangle \left\langle 0\right\vert
+\frac{1}{4}\left\vert 1\right\rangle \left\langle 1\right\vert .\nonumber
\end{align}
In this case, $\left\Vert \rho_{1}-\rho_{2}\right\Vert _{1}=\frac{1}{2},$
$\left\Vert \rho_{1}-\rho_{3}\right\Vert _{1}=\frac{1}{4},$ $\left\Vert
\rho_{2}-\rho_{3}\right\Vert _{1}=\frac{1}{4}$ and%
\begin{align}
\sum_{1\leq i<j\leq3}\left\Vert \rho_{i}-\rho_{j}\right\Vert _{1}  &
=1,\text{ \ \ }\min_{j=1,2,3}\left\{  \sum_{i=1,2,3}\left\Vert \rho_{i}%
-\rho_{j}\right\Vert _{1}\right\}  =\frac{1}{2},\label{52}\\
\max_{j=1,2,3}\left\{  \sum_{i=1,2,3}\left\Vert \rho_{i}-\rho_{j}\right\Vert
_{1}\right\}   &  =\frac{3}{4}.\nonumber
\end{align}
Therefore, in the case considered, the values of the new bounds (\ref{32}),
(\ref{15}), (\ref{17}) are equal to%
\begin{align}
\mathfrak{L}_{1,new}  &  =\frac{19}{48}=0.3958,\text{ \ \ }\mathfrak{L}%
_{2,new}=\frac{7}{18}\simeq0.3889,\label{53}\\
Q_{new}  &  =\frac{4}{9}\simeq0.4444,\nonumber
\end{align}
while the values of the known bounds (\ref{30}), (\ref{40}), (\ref{28.1}) are
\begin{equation}
Q_{4}=\frac{5}{12}\simeq0.4166,\text{ \ \ }Q_{2}=\frac{7}{12}\simeq
0.5833,\text{ \ \ }\mathfrak{L}_{1}=\frac{1}{3}. \label{54}%
\end{equation}

Note that since states (\ref{51}) mutually commute and are of the form
$\rho_{i}=\sum\limits_{n=0,1}\lambda_{n}^{(i)}\left\vert n\right\rangle
\left\langle n\right\vert ,$ $i=1,2,3,$ the optimal success probability
\cite{4}%
\begin{equation}
\mathrm{P}_{\rho_{1},\rho_{2},\rho_{_{3}}|\frac{1}{3},\frac{1}{3},\frac{1}{3}%
}^{opt.success}=\frac{1}{3}\sum\limits_{n=0,1}\left(  \max_{i=1,2,3}%
\lambda_{n}^{(i)}\right)  =\frac{5}{12}=Q_{4}\simeq0.4166. \label{55}%
\end{equation}

For the calculation of the lower bounds (\ref{28.2}), (\ref{28.3}) and the
upper bounds (\ref{41}), (\ref{42}) for equiprobable states (\ref{51}), we
calculate fidelities $F_{ij}:=\left\Vert \sqrt{\rho_{i}}\sqrt{\rho_{j}%
}\right\Vert _{1}$ for states (\ref{51}):%
\begin{align}
F_{12}  &  =\frac{\sqrt{35}+\sqrt{3}}{8}\simeq0.9560,\text{ \ \ }F_{13}%
=\frac{\sqrt{42}+\sqrt{2}}{8}\simeq0,9868,\label{56}\\
F_{23}  &  =\frac{\sqrt{30}+\sqrt{6}}{8}\simeq0,9909,\nonumber
\end{align}
and the trace%
\begin{equation}
\mathrm{tr}\left[  \sqrt{%
{\textstyle\sum\limits_{i=1,2,3}}
\rho_{i}^{2}}\right]  =\frac{\sqrt{110}+\sqrt{14}}{8}\simeq1,7787. \label{57}%
\end{equation}
Therefore,
\begin{align}
Q_{3}  &  =1-\frac{1}{9}\sum_{1\leq i<j\leq3}F_{ij}^{2}\simeq0.6812,\text{
\ \ \ }Q_{5}=\frac{1}{3}\mathrm{tr}\left[  \sqrt{%
{\textstyle\sum\limits_{i=1,2,3}}
\rho_{i}^{2}}\right]  \simeq0,5929,\label{58}\\
\mathfrak{L}_{2}  &  =1-\frac{1}{3}\sum_{1\leq i<j\leq3}F_{ij}\simeq
0,0221,\text{ \ \ \ }\mathfrak{L}_{3}=\text{\ }\left(  Q_{5}\right)
^{2}\simeq0.3515.\nonumber
\end{align}
From (\ref{53})--(\ref{58}) it follows that, in the case considered,%
\begin{equation}
\mathrm{P}_{\rho_{1},\rho_{2},\rho_{_{3}}|\frac{1}{3},\frac{1}{3},\frac{1}{3}%
}^{opt.success}=Q_{4}<Q_{new}<Q_{2}<Q_{5}<Q_{3} \label{59}%
\end{equation}
and
\begin{equation}
\mathrm{P}_{\rho_{1},\rho_{2},\rho_{_{3}}|\frac{1}{3},\frac{1}{3},\frac{1}{3}%
}^{opt.success}>\mathfrak{L}_{1,new}>\mathfrak{L}_{2,new}>\mathfrak{L}%
_{3}>\mathfrak{L}_{1}>\mathfrak{L}_{2}, \label{60}%
\end{equation}
that is, the values of the new lower bounds $\mathfrak{L}_{new,1},$
$\mathfrak{L}_{new,2}$ and the new upper bound $Q_{new}$ are tighter than the
values of the known lower bounds (\ref{28.1})--(\ref{28.3}) and the known
upper bounds (\ref{40})--(\ref{42}), respectively.

\section{Acknowledgments}

The author is grateful to A. S. Holevo and Min Namkung for useful discussions.

\section{Conclusions}

In the present article, for the optimal success probability (\ref{4}), we find
for all $r\geq2:$ (i) the new general lower bounds (Theorem 2) and specify
their relation (Proposition 1) to the general lower bounds known in the
literature; (ii) the new general upper bound (Theorem 3) and specify its
relation (Proposition 2) to the known general upper bounds.

We also present the example where the values of the new general analytical
bounds, lower and upper, on the optimal success probability are tighter than
the values of the known lower bounds (\ref{28.1})--(\ref{28.3}) and the known
upper bounds (\ref{40})--(\ref{42}), respectively.

The new upper bound (\ref{32}) on the optimal success probability has the form
explicitly generalizing to $r>2$ the Helstrom bound in (\ref{5}) and is easily
calculated. For $r=2$, each of our new bounds, lower and upper, reduces to the
Helstrom\ bound in (\ref{5}), and this proves in the other way the Helstrom
result (\ref{6}).

\end{document}